\documentclass{ancs15-alternate}

\usepackage{algorithmic}
\usepackage{algorithm}
\usepackage{verbatim}
\usepackage{amsmath, amssymb}
\usepackage{enumerate}

\usepackage{eufrak}

\usepackage{verbatim}
\usepackage{array}
\usepackage{graphicx}
\usepackage{float}
\usepackage[footnotesize]{subfigure}

\usepackage{balance}

\newtheorem{theorem}{Theorem}

\newfont{\mycrnotice}{ptmr8t at 7pt}
\newfont{\myconfname}{ptmri8t at 7pt}
\let\crnotice\mycrnotice%
\let\confname\myconfname%

\begin{document}

\title{
Enabling Correct Interest Forwarding and Retransmissions in a Content Centric Network} 

\numberofauthors{1}
\author{
    \alignauthor
       J. J. Garcia-Luna-Aceves$^{1,2}$  and Maziar Mirzazad-Barijough$^1$\\
       \affaddr{$^1$Computer Engineering Department, University of California, Santa Cruz, CA 95064}\\
              \affaddr{$^2$PARC, Palo Alto, CA 94304}\\
       \email{\{jj,maziar\}@soe.ucsc.edu}
}
 
\maketitle


\begin{abstract}
We show that the mechanisms used in the name data networking (NDN) and the original content centric networking (CCN)  architectures may not detect Interest loops, even if the network in which they operate is  static and no faults occur. Furthermore, we show that no correct Interest forwarding strategy can be defined that allows Interest aggregation and attempts to
detect Interest looping by identifying Interests uniquely.  We introduce SIFAH (Strategy for Interest  Forwarding and Aggregation with Hop-Counts), the first Interest forwarding strategy shown to be correct under any operational conditions of a content centric network. SIFAH operates by  having forwarding information bases (FIBs) store the next hops and number of hops to named content, and by having each Interest state  the name of the requested content and the hop count from the router forwarding an Interest to the content.   We present the results of simulation experiments using the ndnSIM simulator comparing CCN and NDN with SIFAH. The results of these experiments illustrate the negative impact of undetected Interest looping when Interests are aggregated in CCN and NDN, and the performance advantages of using SIFAH.

\end{abstract}


\vspace{-0.08in}

\vspace{-0.08in}

\vspace{-0.08in}

\section{Introduction}

A number of  information-centric networking (ICN) architectures have been proposed to improve the performance and the end-user experience of the Internet \cite{icn-survey1,  icn-survey2}. ICN architectures focus on (1) enabling access to content and services by name, rather than by original location; (2) protecting content rather than links or connections; and (3) exploiting in-network storage of content. 

A leading approach in ICN architectures can be characterized as {\em Interest-based content-centric networking} and is the focus of this paper.
Directed Diffusion \cite{diffusion} is one of the first examples of this approach.  Requests for named content (called Interests) are diffused throughout a sensor network, and data matching the Interests are sent back to the issuers of Interests. Subsequent proposals (e.g., 
DIRECT \cite{direct})   use a similar approach in MANETs subject to  connectivity disruption. Nodes  use opportunistic caching of content and  flood Interests persistently.  The limitation of Directed Diffusion and other similar  approaches is the need to flood the network with Interests, an approach that cannot be applied at Internet scale.

The original CCN proposal \cite{ccn} was the first example of an Interest-based content-centric  architecture applicable to wired networks in which Interests do not state the identity of the sender. Today, NDN \cite{ndn} and CCN \cite{ccnx} are the leading proposals for  content-centric networking based on Interest forwarding.
In general,  an Interest-based forwarding strategy consists of: populating forwarding information bases (FIB) of routers with routes to name prefixes denoting content, sending content requests (called Interests) for specific named data objects (NDO) over paths implied by the FIBs, and delivering content along the reverse paths traversed by Interests.  

Section~\ref{sec-prev} summarizes  the operation of the  forwarding strategies of  NDN and CCN. The designers of NDN and CCN have argued \cite{ccn, ndn, ndn-fw, ndn-fw2} that an Interest 
stating a name of requested content and a nonce or unique identifier can be forwarded correctly towards an intended node advertising the content name, that routers can aggregate Interests so that a router can forward an Interest for the same content only once, and that
Interest loops can be detected whenever they occur.  However,  no prior work has been reported proving these claims.

Section~\ref{sec-loop} demonstrates that the forwarding strategies of  the original CCN and NDN  architectures \cite{ccn, ndn-fw, ndn-paper} do not work correctly,  in that  some Interests may never return data objects  to the  consumers who issued  the Interests, even if the content does exist in the network, the network topology and routing are stable, and all transmissions are successful. 
More importantly, it is also shown that there is no correct forwarding strategy with Interest aggregation and   Interest-loop detection  based on the matching of  
Interest-identification data  carried in Interests. In this context, 
Interest-identification data can be names of requested content, nonces, unique identifiers, or the path traversed by an Interest.

Section~\ref{sec-design} introduces  the Strategy for Interest  Forwarding and Aggregation with Hop-counts  (SIFAH), which is the first Interest-based forwarding strategy  shown to be correct.
SIFAH operates by  having FIBs store the next hops {\em and} number of hops to named content, and by forwarding each Interest based on  the name of the requested content and a hop count from the forwarding router to the requested content.
A router accepts to forward an Interest only if the hop count stated in the Interest is larger than the hop count from the router to the content as stated in its FIB. Similarly, a router that has forwarded an Interest for a given NDO accepts to aggregate an Interest it receives while waiting for the requested NDO only if the hop count stated in the Interest is larger than the hop count of the Interest  sent by the router.

Section~\ref{sec-correct} proves that SIFAH works correctly when Interest loops occur and Interests are aggregated. 

Section~\ref{sec-perf} analyzes the storage requirements of SIFAH and NDN and shows that  SIFAH is a more desirable approach than using nonces to attempt to detect Interest loops. Furthermore, it presents simulation results based on the unmodified   implementation of the NDN forwarding strategy and  our implementation of SIFAH in ndnSIM. The simulation results 
help to illustrate  that consumers submitting Interests must receive NDO messages or negative acknowledgments (NACK) when SIFAH is used, while 
some Interests may go unanswered in
NDN and the original CCN design due to undetected Interest loops, even in stable topologies with correct entries in FIBs. Furthermore, the results indicate that Interest loops increase the number of PIT entries and end-to-end delays experienced by consumers even when Interest loops are rare.

\section{Existing Interest Forwarding \\ Strategies }
\label{sec-prev}

In NDN and CCN, a given router $r$ uses three primary data structures to implement any of the forwarding strategies defined for  Interest-based content-centric architectures: a forwarding information base ($FIB^r$),  a pending Interest table ($PIT^r$), and a content store ($CS^r$). 

The forwarding strategy determines the  interaction among $FIB^r$, $PIT^r$, and $CS^r$ needed to forward Interests towards nodes advertising having copies of requested content, send NDOs back to consumers who requested them over reverse paths traversed by Interests, and send any other signal indicating the inability to satisfy an Interest. 

$FIB^r$ is used to route incoming Interests to the appropriate next hops towards the desired content producer advertising a content prefix name $n(j)^*$. 

$FIB^r$ is populated using content routing protocols or static routes and matches Interest names stating a specific NDO $n(j)$  to $FIB^r$ entries of prefix names using \emph{longest prefix match}. 

$PIT^r$ serves as a cache of Interest state, such that content objects that satisfy Interests may follow the reverse Interest path back to the original requester. 
$CS^r$ is a cache for content objects.

In the rest of this paper, we use the term name data object (NDO) or content object interchangeably, and  use the term neighbor instead of interface or face.
We denote the name of NDO $j$ by $n(j)$, and the name prefix that includes that NDO name by $n(j)^*$.
We denote the existence of  an entry for a prefix $n(j)^*$ or NDO with name $n(j)$ in the FIB, PIT or CS of router $i$ by $n(j)^* \in FIB^i$, $n(j) \in PIT^i$, and $n(j) \in CS^i$, respectively.

Two Interest-based forwarding strategies proposed to date  are the original CCN strategy \cite{ccn} and the  NDN forwarding strategy \cite{ndn-fw, ndn-paper}. In both strategies,  an Interest created by source $s$ for NDO $j$ states  $n(j)$ and a nonce $id_j(s)$. The pair  $(n(j), id_j(s) )$ is used to denote an Interest uniquely with a large-enough probability. Furthermore, the same pair is used to detect whether an Interest is traversing a loop. 

In the context of NDN and the original CCN, we use $I[n(j), id_j(s) ]$ to denote  an Interest that requests NDO with name $n(j)$  and that is originated by consumer $s$, who  assigns nonce $id_j(s)$ to the Interest.  A content-object message (or NDO message) sent in response to an Interest $I[n(j), id_j(s) ]$, denoted $D[n(j), id_j(s), sig(j) ]$, states the name and nonce of the Interest, a signature payload $sig(j)$ used to validate the content object, and the object itself. 
 
\begin{algorithm}[h]
\caption{NDN Processing of Interest at router $i$}
\label{algo-ndn-Interest}
{\fontsize{8}{8}\selectfont
\begin{algorithmic}[1]
\STATE{{\bf function} Process Interest}
\STATE {\textbf{INPUT:} $PIT^i$,  $CS^i$, $FIB^i$;}
\STATE{\textbf{INPUT:} $I[n(j), id_j(s) ]$ received from $k$; }
\IF{  $n(j) \in CS^i$ }
	\STATE{send $D[n(j), id_j(s), sig(j) ]$ to $k$ }
\ELSE
	\IF{$n(j) \not\in PIT^i$}
		\STATE{
		create  $PI^i_{n(j)}[ id_j(s), in:k, out: \emptyset ]$; \\  
		{\bf call} Forwarding Strategy($PI^i_{n(j)}$) }
	\ELSE 
		\STATE{\% There is a PIT entry for $n(j)$ }
		\IF{$\exists~ PI^i_{n(j)}[ id_j(x) ]$ with  $id_j(x) = id_j(s)$  }
			\STATE{\% A duplicate Interest is detected \\
		{\bf [NDN]}	send $NI[n(j),  id_j(s),  \mathsf{duplicate}]$  to $k$; \\
			drop  $I[n(j), id_j(s) ]$ }
		\ELSE
		\STATE{\% Interest can be aggregated \\
		create $PI^i_{n(j)}[ id_j(s), in:k, out:\emptyset ]$; }
		\IF{$RT_i(I[n(j), id_j(s) ])$ is exprired}
			\STATE{
			{\bf call} Forwarding Strategy($PI^i_{n(j)}$); \\ 
			 }
		\ENDIF
		\ENDIF

	\ENDIF
\ENDIF
\end{algorithmic}}
\end{algorithm}

 \vspace{-0.2in}
\begin{algorithm}[h]
\caption{NDN forwarding of Interest at router $i$}
\label{algo-ndn-fw}
{\fontsize{8}{8}\selectfont
\begin{algorithmic}[1]
\STATE{{\bf function} Forwarding Strategy}
\STATE {\textbf{INPUT:} $PIT^i$,  $CS^i$, $FIB^i$;}
\STATE{\textbf{INPUT:} $PI^i_{n(j)}[ id_j(s), in:k, out: OUTSET ]$ }

\IF{$n(j)^* \in FIB^i $}
	\FOR{{\bf each} neighbor $m$ in $FIB^i_{n(j)^*}$ {\bf by rank}}
		\IF{$m \not=  in:k$ {\bf for all} $ in:k \in PI^i_{n(j)} \wedge$  \\
		      $~~~m \not\in SET $   {\bf for all} $ out:SET \in PI^i_{n(j)} $}
			\IF{$m$ is available}
			\STATE{
			$OUTSET (PI^i_{n(j)}) = OUTSET (PI^i_{n(j)}) \cup m$; \\
			start $RT_i(I[n(j), id_j(s) ])$; \\
			forward $I[n(j), id_j(s) ]$ to neighbor $m$; \\
			{\bf return}
			}
			\ENDIF
		\ENDIF
	\ENDFOR
		\STATE{ 
		{\bf [NDN]} send $NI[n(j), id_j(s),  \mathsf{congestion}]$  to $k$;\\
	 	drop  $I[n(j), id_j(s) ]; $ 
		delete $PI^i_{n(j)}$
		}
\ELSE
	\STATE{
	send $NI[n(j), $ $ id_j(s),  \mathsf{no~ data}]$  to $k$; \\
	 drop  $I[n(j), id_j(s) ]$; 
	 delete $PI^i_{n(j)}$
	  }
\ENDIF
\end{algorithmic}}
\end{algorithm}

The entry in $FIB^i$ for  name prefix $n(j)^*$ is denoted by $FIB^i_{n(j)^*}$ and consists of $n(j)^*$ and the list of neighbors that can be used to reach the NDO. 
If neighbor $k$ is listed in $FIB^i_{n(j)^*}$, then we state $k \in FIB^i_{n(j)^*}$. In NDN \cite{ndn-fw2}, the FIB entry for an NDO also contains a stale time after which the entry could be deleted; the round-trip time through the neighbor;  a rate limit; and status information stating whether it is known or unknown that the neighbor can bring data back, or is known that the neighbor cannot bring data back. 

The entry in $PIT^i$ for NDO with name $n(j)$ is denoted by $PI^i_{n(j)}$ and consists of a vector of one or multiple tuples, one for each nonce processed for the same NDO name.
The tuple for a given NDO  states the nonce used, the incoming and the outgoing neighbor(s). The tuple created as a result of processing Interest $I[n(j), id_j(s) ]$ received from $k$ and forwarded to a set of neighbors $OUTSET$  is denoted by  
$PI^i_{n(j)}[ id_j(s), in:k, out:OUTSET ]$, and the set of outgoing neighbors 
in $PI^i_{n(j)}$ is denoted by $OUTSET(PI^i_{n(j) )}$.

Each PIT entry $PI^i_{n(j)}[ id_j(s), in:k, out:OUTSET ]$ has a lifetime,
which should be larger than the estimated round-trip time to a site where the requested NDO can be found.

We denote by $NI[n(j), id_j(s), \mathsf{CODE}]$ the NACK sent in response to  $I[n(j), id_j(s) ]$, where $\mathsf{CODE}$ states the reason why the NACK is sent. 

Algorithms~\ref{algo-ndn-Interest}  and \ref{algo-ndn-fw} illustrate the NDN Interest processing  approach \cite{ndn-fw, ndn-fw2} using the notation we have introduced, and  correspond to Interest-processing and forwarding-strategy algorithms in \cite{ndn-fw2}. Algorithm \ref{algo-ndn-fw} does not include the probing of neighbors proposed in NDN, given that this aspect of  NDN is  still being defined \cite{ndn-fw2}. Routers forward NACKs received from those neighbors to whom they sent Interests, unless the PIT entries have expired or do not match the information provided in the NACKs. 
The NDN forwarding strategy augments  the original CCN strategy by  introducing  negative acknowledgements (NACK) sent in 
response to Interests for a number of reasons, including:  routers identifying
congestion, routers not having routes in their FIBs to the requested content, or Interest loops being detected.
Algorithms 1 and 2 indicate the use of NACKs that is not  part of the original CCN design by ``{\bf [NDN]}."

\section{Undetected Interest Loops \\ in CCN an NDN  }
\label{sec-loop}

 The use of nonces in NDN and the original CCN approach can be extrapolated  to include the case in which 
an Interest states a nonce and the path traversed by the Interest by assuming that $ id_j(s)$ equals the tuple $( id_j(s)[nonce], $ $id_j(s)[path] )$.
If a nonce and path traversed by the Interest are used, deciding whether an Interest has not traversed a loop can be based on whether $ id_j(x)[nonce] \not= id_j(s)[nonce] \vee i \not\in id_j(s)[path]$. However, including path information in Interests reveals the identity of originators of Interests.

The key aspect of the  forwarding strategies that have been proposed  for NDN and CCN is that  a router determines whether or not an Interest is  a duplicate Interest  based solely on the content name and Interest-identification data for the Interest (a nonce in NDN's case).  To discuss  the correctness of the  forwarding strategy and other strategies,  we define an Interest loop as follows.

\vspace{0.1in}
{\bf  Interest Loop:}  
An Interest loop of $h$ hops for NDO with name  $n(j)$ occurs  when one or more Interests asking for $n(j)$ are forwarded and aggregated by routers along a cycle $L = $   $\{ v_1 , v_2 , ..., v_h , v_1 \}$
such that router $v_k$ receives an Interest for NDO $n(j)$ from $v_{k-1}$ while waiting for a response to the Interest it has forwarded to $v_{k+1}$ for the same NDO, with $1 \leq k \leq h$, $v_{h+1} = v_1$, and $v_{0} = v_h$. $\square$

\vspace{0.1in}
According to the NDN forwarding strategy, a router can select a  neighbor 
to forward an Interest if it is known that it can bring content 
and its performance is ranked higher than other neighbors that can also bring content. The ranking of  neighbors  is done by a router independently of other routers, which can result in long-term routing loops implied by the FIBs if the routing protocol used in the control plane does not guarantee instantaneous loop freedom (e.g., NLSR \cite{nlsr}).

\vspace{-0.15in}
\begin{figure}[h]
\begin{centering}
    \mbox{
    \subfigure{\scalebox{.22}{\includegraphics{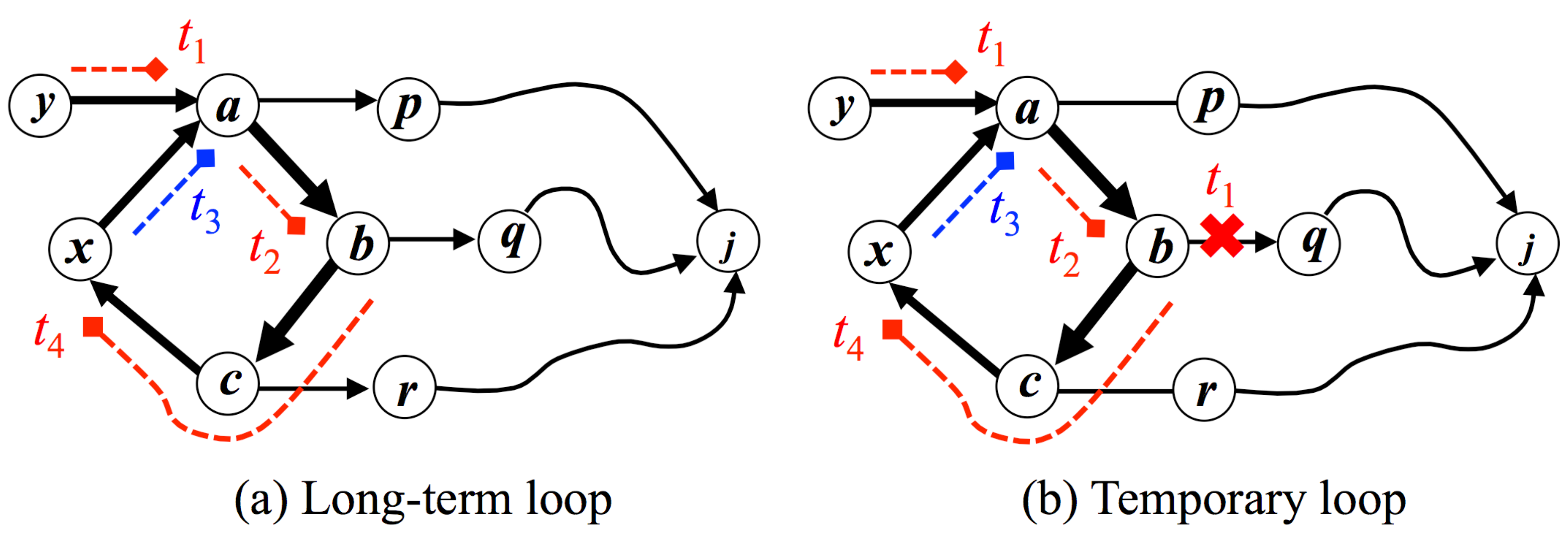}}}
      }
\vspace{-0.15in}
   \caption{Undetected Interest loops in NDN and CCN
   }
   \label{ndn-loop}
\end{centering} 
\end{figure}

Figure \ref{ndn-loop}  illustrates Interest looping in NDN. 
Arrowheads in the figure indicate the next hops to content advertised by router $j$  according to the FIB entries stored in routers.  Thick lines indicate that the perceived performance of a neighbor is 
better than neighbors shown with thinner lines.   Dashed lines indicate the traversal of Interests over links and paths. The time when an event is processed at a router is indicated by $t_i$.
Figure~\ref{ndn-loop}(a) shows the case of a long-term Interest loop formed because the  multi-paths implied in FIBs are not loop-free, even though all routing tables are consistent. 
Figure \ref{ndn-loop}(b) shows the case of a temporary Interest loop when single-path routing is used and FIBs are inconsistent due to a topology change at time $t_1$ (link $(b, q)$ fails). In both cases, router $a$ aggregates  the Interest from $x$ at time $t_3$,  router $x$ aggregates the Interest from $c$ at time $t_4$, and the combined steps preclude the detection of Interest looping. This results in $x$ and $y$ having to wait for their Interests to time out, before they can retransmit. Furthermore, there is no guarantee that their retransmissions will elicit a response (content or NACK).

As Theorem \ref{theo1} proves,  the CCN and NDN forwarding strategies specified in \cite{ccn, ndn-fw2, ndn-paper} 
cannot ensure that Interest loops are detected  when Interests are aggregated, even if nonces were to denote Interests uniquely. The theorem assumes that all messages are sent correctly and that no routing-table changes occur  to show  that the NDN forwarding strategy can fail to return any content or NACK in response to Interests independently of network dynamics.  Furthermore, Theorem \ref{theo2} shows that no forwarding strategy can be correct if it allows Interest aggregation and attempts  Interest-loop detection by the matching of Interest-identification data.

\begin{theorem}
\label{theo1}
Interest loops can go undetected in a stable, error-free network in which  NDN or CCN  is used, even if nonces were to denote Interests uniquely.
\end{theorem}

\begin{proof}
Consider the NDN or CCN forwarding strategy running in a network in which 
no two nonces created by different nodes for the same content are equal, 
all transmissions are received correctly,  and no topology or routing-table changes occur after time $t_0$. 
Let $LT^{v_k}(I[n(j), id_j(s) ])$ denote the lifetime of $I[n(j), id_j(s) ]$ at router $v_k$. 

Assume that  Interests may traverse loops when they are forwarded according to the forwarding strategy, and
let a loop $L = $   $\{ v_1 , v_2 , ..., v_h , v_1 \}$  exist for NDO $j$, and let  Interest $I[n(j), id_j(x) ]$ start traversing the chain of nodes  
$\{v_1 ,$  $v_2 , ...,$ $v_k \}  \in L$  (with $1 < k < h$) at time $t_1 > t_0$.

Assume that $I[n(j), id_j(x) ]$ reaches router $v_k$ at time $t_3 > t_1$ and that router $v_k$ forwards Interest $I[n(j), id_j(y) ]$ to its next hop  $v_{k+1} \in L$ at time $t_2$, where $t_1 \leq t_2 < t_3$,  $id_j(x) \not= id_j(y)$, and $v_{k+1}$ may be $v_1$. 

According to the Interest processing strategy in NDN and CCN, router $v_k$ creates an entry in its PIT for $I[n(j), id_j(y) ]$ at time $t_2$, and perceives any Interest for name $n(j)$ and a  nonce different than $id_j(y)$ received after time $t_2$, and before its PIT entry 
for $I[n(j), id_j(y) ]$ is erased, as a subsequent Interest.

Let $| t_2 -  t_3| < LT^{v_k}(I[n(j), id_j(y) ])$ when router $v_k$ receives $I[n(j), id_j(x) ]$ from router $v_{k - 1} \in L$  at time $t_3$, where $1 < k - 1$. 
According to the  Interest processing strategy in NDN and CCN, router $v_k$  must treat 
$I[n(j), id_j(x) ]$ as  a subsequent Interest for content $n(j)$ that is aggregated, 
because $v_k$ is waiting for $D[n(j), id_j(y) ]$ at time $t_3$.

Because of   the existence of $L$, Interest $I[n(j), id_j(y) ]$ must be forwarded from $v_k$ to $v_1$. Let $t_4$ denote the time when $I[n(j), id_j(y) ]$  reaches $v_1$, where 
$t_4 > t_2 \geq t_1$, and  assume that $| t_1 -  t_4| < LT^{v_1}(I[n(j), id_j(x) ])$. 
According to NDN's Interest processing strategy, $v_1$ must treat $I[n(j), id_j(y) ]$ as a subsequent Interest, because it is waiting for $D[n(j), id_j(x) ]$ at time $t_4$.

Given the Interest aggregation carried out by nodes $v_k$ and $v_1$,  nodes in the chain $\{ v_1 , v_2 , ..., v_{k - 1} \} \in L$ process only 
$I[n(j), id_j(x) ]$,  nodes in the chain $\{ v_{k + 1} , $ $v_{k + 2} , ...,$ $ v_{h} \} \in L$ process only $I[n(j), id_j(y) ]$, and no Interest loop detection can take place. 
Therefore,  no content can be submitted in response to $I[n(j), id_j(x) ]$ and $I[n(j), id_j(y) ]$.
\end{proof}

Similar results to Theorem 1 can be proven for NDN and the original CCN  operating in a network  in which routing tables are inconsistent as a result of network or content dynamics.   In this case, Interest loops can go undetected even if the control plane supports only single-path forwarding of Interests. 

\begin{theorem}
\label{theo2}
No correct forwarding strategy exists with Interest aggregation and  Interest loop detection based on the matching of Interest-identification data.
\end{theorem}

\begin{proof}
Assume any  forwarding strategy in which a router remembers an Interest it has forwarded as long as necessary to detect Interest loops, and detects the occurrence of an Interest loop by matching the Interest-identification data carried in an Interest it receives with the Interest-identification data used in the Interest it forwarded previously asking for the same content. Let 
$I[n(j), id_j(s) ]$ denote the Interest asking for $n(j)$ with Interest-identification data  $id_j(s)$ created by router $s$.

Assume that an Interest  loop $L = $   $\{ v_1 , v_2 , ..., v_h , v_1 \}$ for NDO with name $n(j)$ exists in a network using the  forwarding strategy.
Let Interest $I[n(j), id_j(x) ]$ start traversing the chain of nodes  
$\{v_1 ,$  $v_2 , ...,$ $v_k \}  \in L$  (with $1 < k < h$) at time $t_1$.

Assume that $I[n(j), id_j(x) ]$ reaches router $v_k$ at time $t_3 > t_1$ and that router $v_k$ forwards Interest $I[n(j), id_j(y) ]$ to its next hop  $v_{k+1} \in L$ at time $t_2$, where $t_1 \leq t_2 < t_3$,  $id_j(x) \not= id_j(y)$.
Let $I[n(j), id_j(y) ]$ traverse the chain of nodes $\{v_k ,$  $v_{k+1} , ...,$ $v_1 \}  \in L$, reaching $v_1$ at time $t_4 $, where 
$t_4 > t_2 \geq t_1$.

By assumption, Interest aggregation occurs, and hence $v_k$ aggregates $I[n(j), id_j(x) ]$ at time $t_3$, and $v_1$ aggregates $I[n(j), id_j(y) ]$ at time $t_4$.
Therefore, independently of the amount of information contained in $id_j(x)$ and $id_j(y)$,  $v_1$ cannot receive $I[n(j), id_j(x) ]$ from $v_h$ and 
$v_k$ cannot receive  $I[n(j), id_j(y) ]$ from $v_{k-1}$.
It thus follows that no node in $L$ can successfully 
use the matching of Interest-identification data to detect that 
Interests for $n(j)$ are being sent and aggregated along $L$ and the theorem is true.
\end{proof}

The results in Theorems 1 and 2
can also be proven by mapping the Interest processing strategy of NDN, and any forwarding strategy that attempts to detect Interest loops 
by matching Interest-identification data,  to the problem of distributed termination detection over a cycle, where Interests serve as the tokens of the algorithm \cite{dijkstra, dtd}.
Because Interest aggregation erases a token traversing the ring (Interest loop) when any node in the ring has previously created a different token,  correct termination detection over the ring (i.e., Interest loop detection) cannot be guaranteed in the presence of Interest aggregation.

Obviously,  a loop  traversed by an Interest can be detected easily if each Interest is identified with the route it should traverse.  
This is easy to implement but  requires routers in the network to have  complete topology information
(e.g., \cite{nlsr, icnp98, vutukury}) or at least path information or partial topology information (e.g., \cite{alva, icnp98}).
Similarly, carrying the path traversed by an Interest in its header also ensures that an Interest loop is detected if it occurs.
In these two cases, however, there is no need for using nonces to detect Interest  loops. More importantly, path information reveals the identity of the source router requesting content and hence defeats one of the key objectives of the NDN and CCN forwarding strategies.

Another view of the problem would be to say that Interest aggregation is not common and hence undetected Interest loops should be too rare to cause major performance problems. However, if Interests need not be aggregated, then very different architectures could be designed for content-centric networking that do not require using PITs.

\section{SIFAH
}
\label{sec-design}

\subsection{Design Rationale}

It is clear from the results in the previous section that using nonces or identifying Interests uniquely is useless for Interest-loop detection when Interests are aggregated, and that source routing of Interests  or including the path traversed by an Interest are not desirable. Accordingly, for an Interest  forwarding strategy to be correct in the presence of Interest aggregation, it must be the case that, independently of the identity of an Interest or how Interests for the same content are  aggregated, at least one router detects that it is  traversing a path that  is not  getting the Interest closer to a node that has advertised the requested content. 

Ensuring that at least one router in an Interest loop detects the incorrect forwarding of the Interest  can be attained if Interests were to carry  any type of ordering information that cannot be erased by the use of Interest aggregation.
Fortunately, distance information for advertised name prefixes is exactly this type of ordering information. 

Given that 
forwarding information bases (FIB)  are populated from the routing tables maintained in the control plane of a network, they  constitute  
a readily-available tool to establish the proper interaction between the forwarding strategy operating in the data plane and the distances to advertised content prefixes maintained by the routing protocol  operating in the control plane.  This is the basis of the 
{\em Strategy for Interest  Forwarding and Aggregation with Hop-Counts} (SIFAH).

\subsection{Information Stored and Exchanged}

A router maintains a FIB, a PIT, and an optional  content store.
$FIB^i$ is indexed using content name prefixes. The  FIB entry for  prefix $n(j)^*$ is denoted by $FIB^i_{n(j)^*}$, and consists of a list of one or more  tuples. Each tuple 
states a next hop to $n(j)^*$ and a hop count to the prefix.
The set of next hops to $n(j)^*$ listed in $FIB^i_{n(j)^*}$ is denoted by
$S^i_{n(j)^*}$.
The hop count to $n(j)^*$ through neighbor $q \in S^i_{n(j)^*}$ is denoted by  $h(i, n(j)^*, q)$. 

An Interest sent by node $k$ requesting NDO $n(j)$ is denoted by  $I[n(j), h^I (k) ]$, and states the name $n(j)$, and the hop count ($h^I (k)$) from node  $k$ to the name prefix $n(j)^*$ that is the best match for NDO name $n(j)$  when $k$ forwards the Interest. 

An NDO message sent in response to the Interest $I[n(j),$ $h^I (k) ]$ is denoted by  $D[n(j), sig(j) ]$, and  states the name  of the Interest,  a  signature payload $sig(j)$ used to validate the content object, and the object  itself. 

The NACK sent by router $i$ in response to  an Interest 
is denoted by $NI[n(j),  \mathsf{CODE}]$  where $\mathsf{CODE}$ states the reason why the NACK is sent.  Possible reasons for sending a NACK include: (a) an Interest loop is detected, (b) a  route failed towards the requested content, (c) no content is found, and (d) the PIT entry expired.

$PIT^i$ is indexed using NDO names.
$PI^i_{n(j)}$  denotes the entry created in $PIT^i$ for NDO with name $n(j)$, and specifies: the name of the NDO; 
the hop count $h^I (i)$ assumed by router $i$ when it forwards Interest $I[n(j), h^I (i) ]$; the set of incoming neighbors from which  Interests for $n(j)$ are received ($INSET(PI^i_{n(j)})$); the set of outgoing neighbor(s) ($OUTSET(PI^i_{n(j)})$) to whom router $i$ forwards its Interest;  
and the remaining lifetime for the Interest ($RT(PI^i_{n(j)})$).

\subsection{Interest Loop Detection}

To define  a correct forwarding strategy, special attention must be paid to the fact that updates made to the FIBs stored at routers occur independently of and concurrently with the updates made to their PITs. For example, once a router has forwarded an Interest that assumed a given  distance to content prefix $n(i)^*$ and waits for its Interest to return a data object,
its distance to the same content may change based on updated to its FIB. Hence, simply comparing the minimum distance from a router to content against a distance to content stated in an Interest is not enough to ensure that Interests are not incorrectly forwarded to routers that are farther away form the requested content.

SIFAH takes into account the fact that FIBs and PITs are updated independently  by requiring that a router that forwards an Interest for a given piece of content remembers in its PIT entry the value of the distance to content  assumed 
when it issues its Interest.  The  following rule is then used for a given router to determine whether  an Interest may be propagating over an Interest loop. 

The number of hops to requested content is used as the  
metric  for the invariant condition. This is done for two reasons, 
storing hop-count distances in the FIB incurs less storage overhead than storing complex distance values, and the next hops to a prefix stored in the FIB can be ranked based on the actual distances to content.

\vspace{0.1in}
{\bf HFAR--Hop-Count  Forwarding with Aggregation Rule:}  Router $i$ can 
accept $I[n(j), h^I(k)]$ from router $k$ if one of the following two conditions 
is satisfied:
\begin{enumerate}
\item
$n(j) \not\in PIT^i \wedge \exists ~v  (~ v \in  S^i_{n(j)^*}  \wedge h^I(k)  > h(i, n(j)^*, v) ~)$ 
\item
$n(j) \in PIT^i \wedge h^I(k)  > h^I(i)$
\end{enumerate}

The first condition ensures that router $i$ accepts an Interest from neighbor $k$ only if $i$ determines that is closer to $n(j)^*$ 
through at least one neighbor than $k$ was when it sent its Interest.
The second condition ensures that router $i$ accepts an Interest from neighbor $k$ only if $i$ was closer to  $n(j)^*$  than $k$ when $i$ and $k$ sent their Interests. 

Section~\ref{sec-correct} proves that using HFAR is {\em sufficient} to ensure that an Interest loop  cannot occur without a router in the loop detecting that the Interest has been forwarded incorrectly. This result is independent of whether Interests are aggregated or sent over one or multiple paths, or  how Interests are retransmitted.

Similar  forwarding rules based on more sophisticated lexicographic orderings
could be defined based on the same general approach  stated in HFAR.
The requirement for such forwarding rules is that more information needs to be maintained in the FIBs, such as distance values to name prefixes that take into account such factors as end-to-end delay, reliability, cost, or bandwidth available. 

HFAR is very similar to sufficient conditions for loop-free routing introduced in the past, in particular sufficient conditions for loop-free routing based on diffusing computations \cite{dual, vutukury, dual-patent}. Indeed,  the approach we introduce  for Interest-loop detection in SIFAH
can be viewed as a case of  termination detection based on  diffusing computations \cite{diffuse}.

It should be pointed out   that, because HFAR is not {\em necessary} to detect loops, there are cases in which HFAR  is not satisfied even though no Interest loops exist.  However, prior results on  multi-path routing based on 
diffusing computations \cite{dual2} indicate that this does not constitute a performance problem.
Given that FIBs are updated to reflect correct hop counts, or correct complex distance values in general,
a sufficient condition for loop detection operating with multi-path routing is a good 
baseline for an Interest-based forwarding strategy.

\subsection{SIFAH Operation}

Algorithms~\ref{algo-SIFAH-Interest}  to \ref{algo-SIFAH-link} specify 
the steps taken by  routers to process Interests, forward Interests, return NDOs, process perceived link failures, handle Interest-lifetime expirations, and 
send NACKs according to SIFAH.  Optional steps  and data in algorithms are indicated by ``{\bf [o]}".

The algorithms used to describe SIFAH were not designed to  take into account such issues  as load balancing of available paths, congestion-control, or the forwarding of an Interest over multiple concurrent paths. 
For simplicity,  it is assumed that all Interest retransmissions are carried out on an end-to-end basis (i.e., by the consumers of content) rather than routers. Hence, routers do not attempt to provide any ``local repair" when a neighbor fails or a NACK to an Interest is received; 
the origin of an Interest is in charge of retransmitting it after receiving a NACK for any reason.
Interest retransmissions could also be done by routers.  
The design and analysis of Interest retransmission strategies implemented by routers or by content consumers  is a topic deserving further study. 

{\bf Algorithm~\ref{algo-SIFAH-Interest} } implements HFAR. Router $i$ determines that an Interest can be forwarded because  Condition 1 in HFAR is satisfied (Line 9 of Algorithm~\ref{algo-SIFAH-Interest}), or an Interest can be aggregated because Condition 2 of HFAR is  satisfied (Line 17 of Algorithm~\ref{algo-SIFAH-Interest}). 
Content requests from local content consumers are sent to the router in the form of Interests stating infinite hop counts to content, and  each router knows which neighbors are remote  
and which are local.

\begin{algorithm}[h]
\caption{SIFAH Processing of Interest at router $i$}
\label{algo-SIFAH-Interest}
{\fontsize{8}{8}\selectfont
\begin{algorithmic}[1]
\STATE{{\bf function} Process Interest}
\STATE {\textbf{INPUT:} $PIT^i$,  $CS^i$, $FIB^i$, 
$I[n(j), h^I(k) ]$;}

\STATE{{\bf if}  $n(j) \in CS^i$ {\bf then} send $D[n(j), sig(j) ]$ to $k$ }

\IF{  $n(j) \not\in CS^i$ }
	\IF{$n(j) \not\in PIT^i$}
		\IF{$n(j)^* \not\in FIB^i$}
			\STATE{\% Route failed for $n(j)^*$: \\
			send $NI[n(j),  \mathsf{no~ route}  ]$ to $k$; 
	 		drop $I[n(j), h^I(k) ]$}
		\ELSE 
			\IF{$ \exists ~v \in  S^i_{n(j)^*} (~ h^I(k)  > h(i, n(j)^*, v) ~)$}
             	  \STATE{\% Interest can be forwarded: \\
				{\bf call}  Forwarding Strategy($PI^i_{n(j)}$)
				}				
          		\ELSE
				\STATE{
				\% Interest may be traversing a loop: \\
				send $NI[n(j),  \mathsf{loop} ]$ to $k$; 
          			~~drop $I[n(j), h^I(k) ]$
          			}
			\ENDIF
		\ENDIF
	\ELSE 
		\STATE{\% There is a PIT entry for $n(j)$: }
		\IF{ $h^I(k)  > h^I(i)$ }
          		\STATE{
			\% Interest can be aggregated:\\
			$INSET(PI^i_{n(j)}) = INSET(PI^i_{n(j)}) \cup k$
			}	
		\ELSE
			\STATE{
			\% Interest may be traversing a loop: \\
			send $NI[n(j),  \mathsf{loop} ]$ to $k$; 
          		~drop $I[n(j), h^I(k) ]$
          		}	
		\ENDIF
	\ENDIF
\ENDIF
\STATE{{\bf end function} }
\end{algorithmic}}
\end{algorithm}

The Maximum Interest Life-time ($MIL$) assumed by a router  before it deletes an Interest from its PIT 
should be  large enough to preclude an excessive number of retransmissions.
On the other hand,  $MIL$ should not be too large to cause the PITs to store too many Interests for which no NDO messages or NACKs will be sent due to failures or transmission errors. 
A few seconds would be a viable value for $MIL$.
In practice, however, the consumer submitting an Interest to its local router could provide an initial value for the Interest lifetime estimated over a number of Interests submitted for NDOs in the same NDO group corresponding to a large piece of content (e.g., a movie).  This is specially the case given our assumption that Interest retransmissions are carried out by content consumers, rather than by routers.

{\bf Algorithm~\ref{algo-SIFAH-fw} } describes a simple forwarding strategy 
in which router $i$ simply selects the 
first  neighbor $v$
in the ranked list of neighbors stored in the FIB for prefix $n(j)^*$
that satisfies the first condition in HFAR (Line 4 of the algorithm).
More sophisticated strategies can be devised that attain load balancing among multiple available  routes towards content and can be close to optimum (e.g., \cite{vutukury}). In addition, the same Interest could be forwarded over multiple paths concurrently, in which case  content could be sent back  over some or all the paths that the Interest traversed successfully.
To be effective, however, these  approaches should  require the adoption of a loop-free  multi-path routing protocol in the control plane (e.g., \cite{dcr, dcr-mcast}). In this context, the control plane establishes valid multi-paths to content prefixes using  long-term performance measures, and the data plane exploits those paths using HFAR and short-term performance measurements, without risking the long delays associated with backtracking due to looping. 

\begin{algorithm}[h]
\caption{SIFAH  Interest forwarding at router $i$}
\label{algo-SIFAH-fw}
{\fontsize{8}{8}\selectfont
\begin{algorithmic}[1]
\STATE{{\bf function} Forwarding Strategy}
\STATE {\textbf{INPUT:} $PIT^i$,  
$FIB^i$, $MIL$, $I[n(j), h^I(k) ]$;}

		\FOR{{\bf each} $v \in S^i_{n(j)^*}$ {\bf by rank}} 
			\IF{
			$h^I(k) > h(i, n(j)^*, v) $}
				\STATE{
				create $PI^i_{n(j)}$; \\
				$INSET(PI^i_{n(j)}) = \{k\}$; 
				$OUTSET(PI^i_{n(j)}) = \{v \}$; \\
				$RT(PI^i_{n(j)}) = MIL$; 
				$h^I(i) = h(i, n(j)^*, v) $; \\
				 forward $I[n(j), h^I(i) ]$ to  $v$;  
				 {\bf return}
				}
			\ENDIF
		\ENDFOR
		\STATE{\% No neighbor can be used in $ S^i_{n(j)^*}$: \\
		{\bf for each} $k \in INSET(PI^i_{n(j)})$ send $NI[n(j),   \mathsf{no~ route}  ]$ to $k$}

\STATE{{\bf end function} }
\end{algorithmic}}
\end{algorithm}

{\bf Algorithm~\ref{algo-SIFAH-Data}} outlines the processing of NDO messages received in response to Interests.  A router accepts an NDO received from a neighbor if it has a PIT entry waiting for the content and the NDO message came from one of the neighbors over which the Interest was sent (Line 5 of the algorithm). 
The router   forwards the valid NDO to any neighbor that requested it and deletes the corresponding PIT entry. 
A router stores an NDO it receives optionally (Step 7 of Algorithm~\ref{algo-SIFAH-Data}).
The caching  of NDOs is done according to the   caching strategy used in the network, which can be path-based or edge-based \cite{caching}, for example. However, SIFAH works independently of the caching strategy adopted in the network.

\begin{algorithm}[h]
\caption{Process NDO message from $q$ at router $i$}
\label{algo-SIFAH-Data}
{\fontsize{8}{8}\selectfont
\begin{algorithmic}[1]
\STATE{{\bf function} Process NDO message}
\STATE{\textbf{INPUT:} $PIT^i$,  $CS^i$,  $FIB^i$, $D[n(j), sig(j) ]$ received from  $q$; }
\STATE{{\bf [o]} verify $ sig(j)$;}
\STATE{{\bf [o]}  {\bf if} verification fails {\bf then} drop $D[n(j), sig(j) ]$}
\IF{$n(j) \in PIT^i \wedge q \in OUTSET(PI^i_{n(j)})$}
	\STATE{{\bf for each} $p \in INSET(PI^i_{n(j)})$ {\bf do} \\send $D[n(j), sig(j) ]$ to $p$;}
	\STATE{{\bf [o]}  store the content with name $n(j)$ in $CS^i$;}

      \STATE{delete $PI^i_{n(j)} $}
  \ELSE
  	\STATE{drop $D[n(j), sig(j) ]$ }

\ENDIF
\STATE{{\bf end function} }
\end{algorithmic}}       
\end{algorithm}

{\bf  Algorithm~\ref{algo-SIFAH-timer}} shows a simple approach to handle  the case when a PIT entry expires with no NDO or NACK being received.  
Given that routers do not initiate Interest retransmissions, router $i$ simply sends NACKs to all  neighbors from which it received Interests for $n(j)$. 
A more sophisticated approach would be needed for the case  in which routers must provide Interest retransmissions in a way similar to on-demand routing protocols that support local repair of route requests. 

\begin{algorithm}[h]
\caption{Process Interest life-time expiration }
\label{algo-SIFAH-timer}
{\fontsize{8}{8}\selectfont
\begin{algorithmic}[1]
\STATE{{\bf function} Process Interest Life-time Expiration}
\STATE {\textbf{INPUT:} $PIT^i$,  
$RT(P^i_{n(j)}) = 0$;
}

\STATE{{\bf for each} $p \in INSET(PI^i_{n(j)})$ {\bf do} \\ send $NI[n(j),   \mathsf{Interest ~expired} ]$}
\STATE{delete $PI^i_{n(j)}$}
     
\STATE{{\bf end function} }
\end{algorithmic}}
\end{algorithm}

{\bf Algorithm~\ref{algo-SIFAH-nack}} states the steps taken to handle NACKs.
Router $i$ forwards the NACK it receives for $n(j)$ to all those neighbors from whom it received Interests for $n(j)$ and deletes the Interest entry after that.
Supporting Interest retransmissions by routers would require a more complex approach for the handling of NACKs. 

\begin{algorithm}[h]
\caption{Process NACK  at router $i$}
\label{algo-SIFAH-nack}
{\fontsize{8}{8}\selectfont
\begin{algorithmic}[1]
\STATE{{\bf function} Process NACK}
\STATE {\textbf{INPUT:} $PIT^i$,  
$NI[n(j), \mathsf{CODE}]$; 
}
\IF{$n(j) \not\in PIT^i$ } 
 	\STATE{drop $NI[n(j),  \mathsf{CODE} ]$ }
\ELSE	
	\STATE{{\bf if}  $k \not\in OUTSET(PI^i_{n(j)})$ {\bf then} drop $NI[n(j),  \mathsf{CODE} ]$; }
	\IF{ $k \in OUTSET(PI^i_{n(j)})$}
		\STATE{{\bf for each} $p \in INSET(PI^i_{n(j)})$ {\bf do} \\send $NI[n(j),   \mathsf{CODE} ]$; }

     		\STATE{delete $PI^i_{n(j)}$}

	\ENDIF
\ENDIF
\STATE{{\bf end function} }
\end{algorithmic}}
\end{algorithm}

 \vspace{-0.1in}
\begin{algorithm}[h]
\caption{Process failure of link  $(i, k)$  at router $i$ }
\label{algo-SIFAH-link}
{\fontsize{8}{8}\selectfont
\begin{algorithmic}[1]
\STATE{{\bf function} Process Link Failure}
\STATE {\textbf{INPUT:} $PIT^i$;  
}
\FOR{{\bf each} $n(j) \in PIT(i)$}
	\IF{$k  \in INSET(PI^i_{n(j)})$}
		\STATE{
		$INSET(PI^i_{n(j)}) = INSET(PI^i_{n(j)}) - \{k\}$; \\
		{\bf if} $INSET(PI^i_{n(j)}) = \emptyset$ {\bf then}
		delete $PI^i_{n(j)}$; 
		}
	\ENDIF
	\IF{$k  \in OUTSET(PI^i_{n(j)})$}	
		\STATE{
		$OUTSET(PI^i_{n(j)}) = OUTSET(PI^i_{n(j)}) - \{k\}$; }	
		\IF{$OUTSET(PI^i_{n(j)}) = \emptyset$}
			\FOR{{\bf each}  $p \in INSET(PI^i_{n(j)})$}
       				\STATE{
					send $NI[n(j),  \mathsf{route ~failed} ]$
					}
			\ENDFOR	
			\STATE{delete $PI^i_{n(j)}$}
		
		\ENDIF
	\ENDIF
\ENDFOR

\STATE{{\bf end function} }
\end{algorithmic}}
\end{algorithm}

{\bf Algorithm~\ref{algo-SIFAH-link}} lists the steps taken by a router in response to the failure of connectivity with a neighbor. 
Reacting to the failure of perceived connectivity with a neighbor over which Interests have been forwarded  could be simply to wait for the life-times of those Interests to expire. However, such an approach can be  very slow  reacting to link failures compared to using Algorithm~\ref{algo-SIFAH-link}.
The algorithm assumes that the control plane updates $FIB^i$ to reflect any changes in hop counts  to name prefixes resulting from the loss of connectivity to one or more neighbors. 
For each Interest that was forwarded over the failed link, router $i$ sends a NACK to all   neighbors whose Interests were aggregated.

\subsection{Examples of SIFAH Operation}

Figures~\ref{no-loop}(a) to (d) illustrate how SIFAH operates using the 
same example used in Figure~\ref{ndn-loop}.
Figures~\ref{no-loop}(a) and (b) address the case in which 
the control plane establishes  multiple paths to each name prefix but 
does not guarantee loop-free routing tables.  Figures \ref{no-loop}(c) and (d) illustrate how SIFAH operates when single-path routing is used.

The pair of numbers next to each link outgoing from a node in Figure~\ref{no-loop}(a) 
indicates the hop count to $n(j)$ through 
a neighbor and the ranking of the neighbor in the FIB. The example assumes that: (a) routers execute a routing protocol that does not enforce loop-free  FIBs; and (b) the ranking of neighbors  is determined independently at each router using some data-plane strategy based on the  perceived performance of each path and interface. 
It should be noted that the distance value of a path need not be directly proportional to the hop-count value of the path shown in the figure. 

Let the tuple  ($v$: $h, r$) indicate a neighbor, its hop count and its ranking. In Figure~\ref{no-loop}(a),  $FIB^a$  lists ($b$: 7, 1), ($p$: 7, 2), and 
($x$: 9, 3), which  is  shown in green font. Similarly,  $FIB^y$ states ($a$: 8, 1);
$FIB^b$ states ($c$: 10, 2), ($a$: 8, 1), and  ($q$: 6, 3); $FIB^c$ states ($b$: 7, 1), ($x$: 9, 2), and ($r$: 9, 3); and $FIB^x$ states ($a$: 8, 1) and ($c$: 8, 2). 
Some of the FIB entries for $p$, $q$ and $r$ are shown in  black font.

\vspace{-0.1in}
\begin{figure}[h]
\begin{centering}
    \mbox{
    \subfigure{\scalebox{.24}{\includegraphics{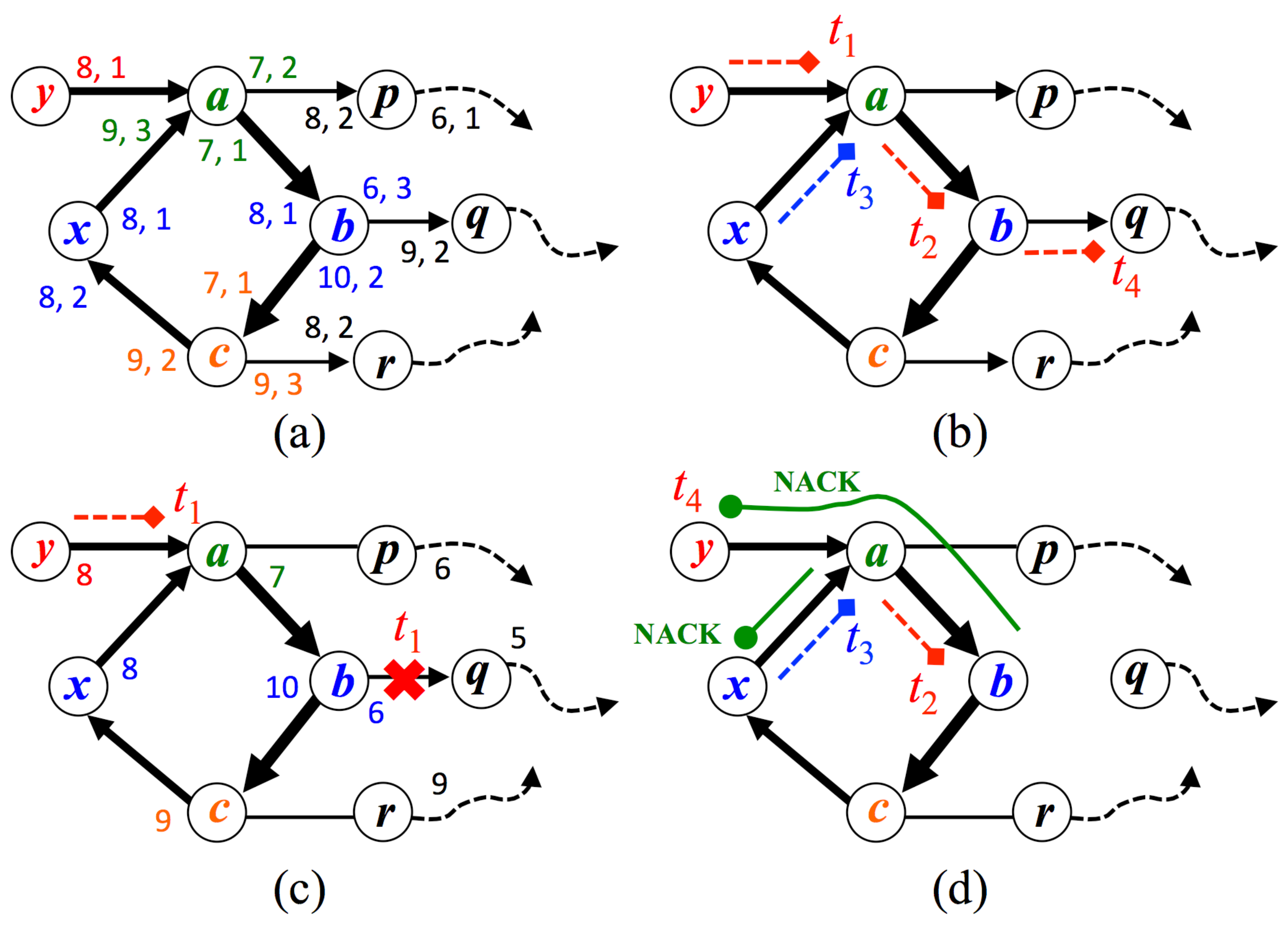}}}
      }
 \vspace{-0.1in}
   \caption{Interest looping is avoided or detected with SIFAH
   }
   \label{no-loop}
\end{centering} 
\end{figure}  

In Figure ~\ref{no-loop}(b),  router $y$ originates an Interest for 
$n(j)$ and sends $I[n(j), h^I(y) = 8]$
to $a$. Router $a$ receives the Interest from router $y$ at time $t_1$ and, 
given that $8 = h^I(y) >$ $h(a, n(j)^*, b) = 7$, 
it accepts the Interest because it has at least one neighbor
that satisfies HFAR. Router $a$ sends  $I[n(j), h^I(a) = 7]$ to $b$ because 
it is the highest-ranked neighbor satisfying HFAR.  Router $a$ aggregates $I[n(j),$ $ h^I(x) = 8]$ at time $t_3 > t_1$, because it sent $I[n(j),$ $ h^I(a) = 7]$ at time $t_1$ and $8 = h^I(x) > h^I(a) = 7$. 
Router $b$ receives the Interest from $a$ at time $t_2 > t_1$;  
accepts it because it has at least one neighbor that satisfies HFAR
($7 = h^I(a) > h(b, n(j)^*, q) = 6$); and sends $I[n(j), h^I(b) = 6]$ to $q$ because  $q$ is the highest-ranked neighbor of $b$ that satisfies HFAR.  
This is an example that Interests are forwarded along loop-free paths 
if SIFAH is used and the FIBs maintained by routers have consistent information, even if some of the multi-paths implied in the FIBs involve loops.
The next section proves this result in the general case.

Figure \ref{no-loop}(c) shows the hop count values stored in the FIBs for name prefix $n(j)$  when single-path routing is used. Each router has a single next hop and one hop count for each prefix listed in its FIB.  
Router $b$ updates its FIB to reflect the failure of link $(b, q)$  at time $t_1$, while router $y$ sends an Interest to router $a$ requesting $n(j)$.
Routers have inconsistent FIB states for $n(j)$ while routing updates propagate and Interests are being forwarded.

As shown in Figure \ref{no-loop}(d), router $b$ {\em must} send  $NI[n(j),  \mathsf{loop}]$ to $a$, 
because $7 = $ $h^I(a) \not> $ $h(b, n(j)^*, c) = 10$ and HFAR is not satisfied. In turn, when $a$ receives the NACK from $b$, it must forward $NI[n(j),  \mathsf{loop}]$  to $y$ and to $x$. Eventually, the routing protocol running in the control plane makes routers $a$ and $y$ change the hop count to $n(j)^*$ in their FIBs to reflect the failure of link $(b, q)$. At that point, a retransmission of the Interest from $y$ would state $h^I(y) = 9$ and would make $a$ forward  $I[n(j), h^I(a) = 8]$ to $p$.

\section{Correctness of SIFAH}
\label{sec-correct}

The following theorems show that SIFAH enforces  correct Interest forwarding and aggregation, and constitutes a safe Interest forwarding strategy.
The results  are independent of whether the network is static or dynamic, the specific caching strategy used in the network (e.g., at the edge or along paths traversed by NDO messages \cite{caching}), or the retransmission strategy used by content consumers after experiencing g a timeout or receiving a NACK from attached routers. SIFAH ensures that Interests cannot be incorrectly propagated and aggregated along  loops without 
meeting routers that detect the incorrect forwarding and hence send  NACKs in return. 

\begin{theorem}
\label{theo3}
Interest loops cannot occur and be undetected in a network in which SIFAH is used.
\end{theorem}

\begin{proof}
Consider a network in which SIFAH is used.  Assume for the sake of contradiction that  nodes in a loop  $L$ of $h$ hops  $\{ v_1 , $ $v_2 , ..., $ $v_h , v_1 \}$  send and possibly aggregate
Interests for $n(j)$ along $L$, with no node in $L$  detecting the incorrect forwarding of any of the Interests sent over the loop.

Given that  $L$ exists by assumption, $v_k \in L$ must send $I[n(j), h^I(v_{k}) ]$ to node $v_{k+1} \in L$ for $1 \leq k \leq h - 1$, and $v_h \in L$ must send $I[n(j), h^I(v_{h}) ]$ to node $v_{1} \in L$. 
For $1 \leq k \leq h - 1$, let $h(v_k, n(j)^*)^L$ denote the value of $h^I (v_k)$ when node 
$v_k$   sends $I[n(j), h^I(v_{k}) ]$  to node $v_{k+1}$,  
with $h(v_k, n(j)^*)^L = h(v_k, n(j)^*, v_{k+1})$. 
Let $h(v_h, n(j)^*)^L$ denote 
the value of $h^I (v_h)$ when when node $v_h$ sends
$I[n(j), h^I(v_{h}) ]$ to node $v_1 \in L$, with $h(v_h, n(j)^*)^L =$  $ h(v_h, n(j)^*, v_1)$. 

Because no node in $L$ detects the incorrect forwarding  of an Interest, each node in $L$ must aggregate the Interest it receives from the previous hop in $L$ or it must send its own Interest as a result of the Interest it receives from the previous hop in $L$. This implies that
$v_k \in L$ must accept $I[n(j), h^I(v_{k-1}) ]$ before $RT(PI^{v_k}_{n(j)})$ expires
for $1 \leq k < h $, and  $v_1 \in L$ must accept $I[n(j), h^I(v_{h}) ]$
before $RT(PI^{v_1}_{n(j)})$ expires.

According to SIFAH, if $v_k$ aggregates $I[n(j),$ $ h^I(v_{k-1}) ]$, then it must be true that $h^I(v_{k-1}) > h^I(v_{k})$. 
Similarly, if $v_1$ aggregates $I[n(j), h^I(v_{h}) ]$, then it must be the case that  $h^I(v_{h}) > h^I(v_{1})$. 

On the other hand, if $v_k$ sends $I[n(j), h^I(v_{k}) ]$ to $v_{k+1}$ as a result of receiving 
$I[n(j), h^I(v_{k-1}) ]$ from $v_{k-1}$, then 
it must be true that  $h^I(v_{k - 1}) > h(v_{k}, n(j)^*)^L = h^I(v_{k}) $ for $1 < k \leq h$.
Similarly, if
$v_1$ sends $I[n(j), h^I(v_{1}) ]$ to $v_{2}$ as a result of receiving 
$I[n(j), h^I(v_{h}) ]$ from  $v_{h}$, then  $h^I(v_{h}) > h(v_{1}, n(j)^*)^L = h^I(v_{1}) $.

It follows from the above argument that, for $L$ to exist when each node in the loop follows SIFAH to send Interests asking for $n(j)$, it must be true that 
$h^I(v_{h}) > h^I(v_{1}) $
and
$h^I(v_{k - 1}) >  h^I(v_{k}) $ for $1 < k \leq h$.
However, this is a contradiction, because it implies that $h^I(v_{k}) >  h^I(v_{k})$ for $1 \leq k \leq h $.
Therefore, the theorem is true.\end{proof}

The proof of Theorem 3 can be augmented to account for Interest forwarding strategies based on complex distance values rather than hop counts.

To be safe, an Interest forwarding strategy must ensure that either an NDO message with the requested content or a NACK is received within a finite time by the consumer who issues an Interest. The following theorem shows that this is the case for SIFAH, independently of the state of the topology or the fate of messages.

\begin{theorem}
\label{theo4}
SIFAH ensures that an NDO message for name $n(j)$  or a NACK  is received within a finite time by any consumer who issues an Interest for NDO with name $n(j)$.
\end{theorem}

\begin{proof}
Consider  
$I[n(j), h^I (s) ]$ being issued by consumer $s$ at time $t_1$. 
The forwarding of Interests assumed in SIFAH is based on the best match of the requested NDO name with the prefixes advertised in the network. Furthermore, according to Algorithm 3, a router sends back an NDO message to a neighbor that sent an Interest for NDO $n(j)$ only if has an exact match of the name $n(j)$ in its content store. According to Algorithm 5, a router that receives an NDO message in response to an Interest it forwarded must forward the same NDO message. Hence, the wrong NDO message cannot be sent in response to an Interest.
There are three cases to consider next: (a) there are no routes to the name prefix $n(j)^*$ of the requested NDO, (b) the Interest  traverses an Interest loop, or (c) the Interest traverses a simple path towards a router $d$ that can reply to the Interest. 

{\em Case 1:} If there is no route to $n(j)^*$, then it follows from the operation of SIFAH (Algorithm 4) that a router issues a NACK stating that there is no route. That NACK is either forwarded successfully back to $s$ or is lost due to errors or faults. In the latter case, it follows from Algorithms 6 and 8 that a router must send a NACK back towards $s$ stating that the Interest expired or the route failed. 

{\em Case 2:}  If $I[n(j), h^I (s) ]$  is forwarded along an Interest loop and does not reach any node with a copy of $n(j)$, then it follows from  Theorem \ref{theo3} that the Interest must either reach some router $k$ that detects the incorrect forwarding of the Interest and must issue a NACK $NI[n(j), \mathsf{loop} ]$ in response, or the Interest is dropped due to faults or transmission errors before reaching such router $k$.

If $NI[n(j), \mathsf{loop} ]$ reaches a router $k$ that detects the loop and issues  $NI[n(j), \mathsf{loop} ]$,  then
according to SIFAH (Algorithm 7), every router receiving the NACK $NI[n(j), \mathsf{loop} ]$ originated by router $k$ from the neighbor to whom the Interest was sent must relay 
the NACK towards $s$. Hence, if no errors or faults prevent the NACK from reaching $s$, the consumer receives a NACK stating that an Interest loop was found.

On the other hand, if either the Interest traversing an Interest loop or the NACK it induces at some router $k$ is lost, it follows from Algorithms 6 and 8 that a router between $s$ and  router $k$ must send a NACK towards $s$ indicating  that the Interest expired or that the route failed. Accordingly,  consumer $s$ must receive a NACK within a finite time after issuing its Interest in this case.

{\em Case 3:} If the Interest  traverses a simple path towards a router $d$ that advertises $n(j)^*$ or has a content store containing $n(j)$, then the Interest must either reach $d$ or not. 

If the Interest is lost and does not reach $d$, then it follows from Algorithms 6 and 8 that a router between $s$ and router $d$  must send a NACK towards $s$ indicating  that the Interest expired or that the route failed. As a result, $s$ must receive a NACK originated by some router between $s$ and $d$.

If the Interest reaches $d$, then that router must either send the requested NDO back, or (in the case that $d$ advertises $n(j)^*$ and $n(j)$ does not exist) issue a NACK stating that 
$n(j)$ does not exist. According to Algorithms 5 and 7, the NDO message or NACK originated by $d$ is forwarded back towards $s$ along the reversed simple path traversed by the Interest.  If no fault or errors occur between $d$ and $s$, it follows that the theorem is true for this case. Alternatively, if the NDO or NACK originated by $d$ is lost due to faults or errors, it follows from  Algorithms 6 and 8 that a router between $s$ and  router $d$ must send a NACK towards $s$ indicating  that the Interest expired or that the route failed. 
\end{proof}

\section{Performance Comparison}
\label{sec-perf}

We compare SIFAH with NDN and the original CCN forwarding strategy in terms of the storage complexity of the approaches;
the average time that a PIT entry remains in the PIT waiting for an NDO message or a NACK to be received in response, which we call PIT entry pending time; the end-to-end delay experienced by content consumers in receiving either the content they request or negative feedback; and the number of entries in the PITs maintained by content routers.

The storage complexity of each approach provides an indication of the storage overhead induced by the type of information required for routers to detect Interests loops. The simulation results we present on  PIT entry pending times, end-to-end delays, and PIT sizes 
should be viewed simply as indications of the negative effects that undetected Interest loops have on the performance of NDN and CCN, and the fact that they can be completely avoided using SIFAH.

\subsection{Storage Complexity}

There is a large difference in the storage overhead incurred with 
the NDN  forwarding strategy compared to SIFAH.

In SIFAH, router $i$ uses only the value of $h^I(i)$ to determine whether the Interest it receives from $k$ may be traversing an Interest loop, and does not store $h^I(k)$. Hence, the PIT storage size 
for SIFAH is  
\[
SS_{SIFAH} = O(( INT + |mh| )  |PIT^i|_{SIFAH})
\]
\noindent
where $|PIT^i |_{SIFAH}$ is the number of pending Interests in $PIT^i$ when SIFAH is used,  $|mh|$ is the number of bits used to store $h^I(i)$, and
$INT$ is the average storage required to maintain information about the incoming and outgoing neighbors for a given Interest.
For a given NDO with name $n(j)$, the amount of storage needed to maintain the incoming and outgoing neighbors is 
\[
INSET(PI^i_{n(j)}) + OUTSET(PI^i_{n(j)}).
\]

The NDN  forwarding strategy requires each router to store the list of different  nonces used to denote valid Interests for a given NDO name $n(j)$. With each nonce being of size $|id|$ and router $i$ having up to $I$ neighbors that send valid Interests for an NDO, the 
PIT storage size for NDN is 
\[SS_{NDN}  = O((INT + |id| I )   ~|PIT^i |_{ NDN})
\]
\noindent
 where $|PIT^i |_{ NDN}$ is the number of pending Interests in $PIT^i$ when NDN is used. 
Hence, even if $|PIT^i |_{ NDN}$ is the same as $ |PIT^i |_{ SIFAH}$, 
the amount of additional PIT storage needed in NDN over SIFAH is 
\begin{eqnarray*}
&&SS_{NDN}- SS_{SIFAH} \geq
\\
& & (|id| I)(|PIT^i |_{ NDN}) -  (|mh|) (|PIT^i |_{ NDN}).
\end{eqnarray*}

A   maximum hop count  of 255 for an Interest
is more than enough. Hence,  with the  size of  a nonce  in NDN of four bytes, the savings in PIT storage obtained with SIFAH  compared to NDN  is
$(32 I - 8)~|PIT^i |_{ NDN}$. This  represents enormous savings   of RAM in large networks.
Furthermore, because the NDN forwarding strategy may not detect loops when Interests are aggregated,  many Interest entries in PITs 
may have to be stored until their lifetimes expire. Accordingly, 
$|PIT^i |_{SIFAH}$ can be much smaller than $|PIT^i |_{ NDN}$. This is confirmed by the simulation results presented subsequently.

The additional FIB storage overhead  in SIFAH compared to the NDN forwarding strategy consists of storing the hop count information for each prefix $n(j)^*$ from each neighbor. This amounts to $ (|mh|)(  |FIB^i|)D^i$ at router $i$, where $D^i$ is the  number of neighbors of router $i$ and 
$|FIB^i|$ is the number of entries in $FIB^i$. Given that $D^i$ and $I$ are of the same order and $O(|FIB^i|) < O(|PIT^i|)$, this  is far smaller than the additional PIT storage needed by the NDN forwarding strategy compared to SIFAH.

\subsection{Performance Impact of \\Undetected Interest Loops }

\subsubsection{ Implementation of Forwarding Strategies \\in ndnSIM}

We implemented SIFAH in ndnSIM,   an open-source NS-3 based simulator for Named Data Networks and Information Centric Networks \cite{ndnSIM}. Following the NDN architecture, ndnSIM is implemented as a new network-layer protocol model, which can run on top of any available link-layer protocol model, as well as on top of network-layer and transport-layer protocols. 

We used the NDN implementation of its data plane from ndnSIM without any modifications.
The ndnSIM NDN implementation  is capable of detecting simple loops by matching nonces.
The PIT entry expiration time for NDN is set to the default of one second.  It should be pointed out that, in the default NDN implementation, a router that receives a duplicate Interest simply drops the Interest without sending a NACK back. This corresponds to the original CCN forwarding strategy. The ndnSIM NDN implementation also allows the use of NACKs after Interest loop detection. The results presented in this section for  ``CCN" correspond to the ndnSIM implementation of NDN without NACKs, and the results presented for  ``NDN" correspond to the ndnSIM implementation of NDN with NACKs enabled. 

To implement Algorithms 3 to 8 defining SIFAH in ndnSIM, we had to make some modifications on the basic structures of ndnSIM, namely:  the FIBs, Interest packets, NACKs, and the forwarding strategy. 
A new field ``rank'' is added to every entry of the FIB. Unlike ndnSIM in which the next hop selection for requested prefixes is based on hop count, in SIFAH next hops are sorted based on rank of each FIB entry.  
The field $h(k)$ was added to each Interest message, which determines the hop count from forwarding node $k$ to the prefix requested by the Interest. 
A new type of NACK for loop detection is added and the behavior of forwarding strategy for NACKs is modified based on SIFAH definitions. Furthermore, a new class of forwarding strategy is added to ndnSIM that implements SIFAH functions.
 
\subsubsection{Simulation Scenarios}

To isolate the operation of the data plane from the performance of different routing protocols operating in the control plane, we used  static routes and  manually configured routing loops for specific prefixes. 

Given the use of static routes and configured loops, we used a simple grid topology of sixteen nodes with two consumers producing Interests with different prefixes and one producer announcing the  content requested in the Interests. Interest traffic is generated at a constant bit rate with a frequency of 2000 Interests per second. The delay over each link of the topology is set t to 10 msec  and PIT entry expiration time is set to only 1000 msec, which is too short for real networks but is large enough to illustrate the consequences of undetected Interest loops.  

Five different scenarios, each lasting  90 seconds of simulation time,
were used to compare SIFAH with NDN and CCN. Each scenario is defined by the percentage of Interests traversing loops, which was set to equal 
0\%, 10\%, 20\%, 50\%, and 100\%  of the Interests generated by consumers.
In practice, it should be the case that  only a small fraction of Interests traverse loops, assuming a correct routing protocol is used in the control plane and sensible policies are used to rank the available routes in the FIBs.
The scenarios we present illustrate that 
just a few Interests traversing undetected loops cause performance degradation, and that network performance is determined by the PIT entry expiration times as the fraction of Interests traveling loops increases.

\begin{figure}[h]
\begin{centering}
    \mbox{
    \subfigure{\scalebox{.28}{\includegraphics{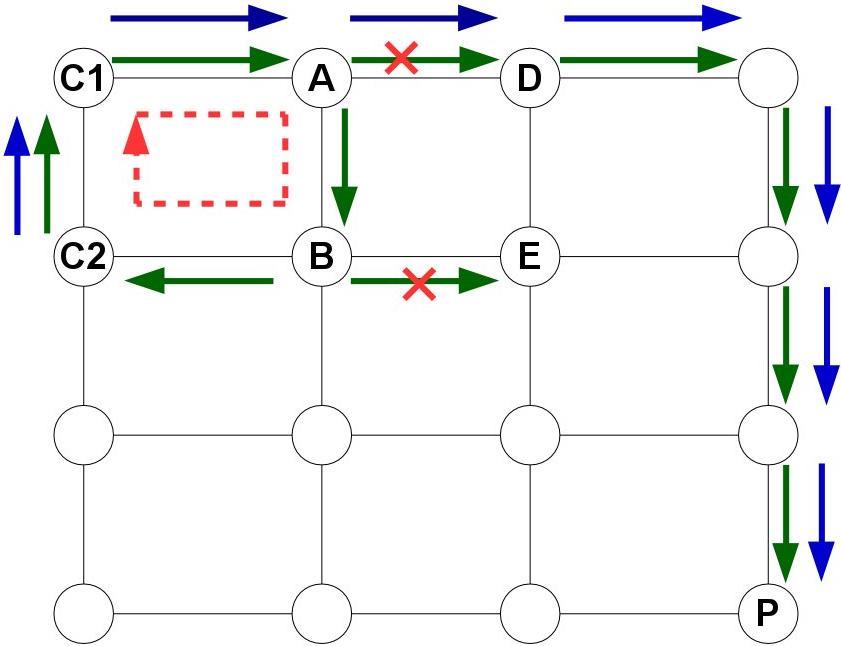}}}
      }
 \vspace{-0.1in}
   \caption{Initial Routes and Custom Loop Scenario}
   \label{loop}
\end{centering} 
\end{figure} 

Figure ~\ref{loop} shows the topology and scenario we used in our simulations. Consumer C1, produces Interests for $prefix_1$ and $prefix_2$, and consumer C2 produces Interests only for $prefix_2$. Blue arrows and green arrows shows initial routes for $prefix_1$ and $prefix_2$, respectively. We assume that the route between nodes A and D, and the route between nodes B and E for $prefix_2$ are disconnected. Therefore, Interests requesting $prefix_2$ use alternate paths from node A to node B and from node B to node C2, which causes the looping of such Interests.

Interests for $prefix_2$ generated by C1 and C2, request the same content at approximately the same time, so that aggregation can take place at routers along the paths traversed by Interests. This results in the aggregation of Interests at node C1 for Interests generated by C2, and the aggregation of Interests  at node C2 for Interests generated by C1. Our simple scenarios provide enough insight on the negative impact of undetected Interest loops in the presence of Interest aggregation using NDN and the original CCN design. 

Simulation results are shown for three different forwarding strategies: The original CCN, NDN, and SIFAH. The
difference between CCN and NDN is that CCN  does not send NACKs when duplicate Interests are detected. On the other hand, NDN sends NACK when simple loops are detected by receiving duplicate Interests.

\subsubsection{Impact on PIT Entry Duration}

Figure ~\ref{pending-time} shows the average value of the PIT entry pending time for all PIT entries. When no Interest loops are present,  NDN, CCN and SIFAH 
exhibit the same performance, with each having an average PIT entry pending time of 60 msec.  This should be expected, given that Interests and NDO messages traverse shortest paths between consumers and producers or caches.

The average PIT entry pending time in SIFAH does not increase  as  he percentage of Interests that 
encounter  Interest loops  increases. 
The reason  for this  is that  SIFAH ensures that an  Interest must elicit either an NDO message or a NACK to be sent back from some router along the route it traverses back to the consumer that originates the Interest. Hence, the average amount of time an Interest entry spends in the PIT is a function of the round-trip time it takes for either an NDO message or a NACK to evict it from the PIT. This is proportional to a round-trip time between a consumer and a router with the content or a router  at which HFAR is not satisfied, which is a few milliseconds in the grid topology.

\vspace{-0.1in}
\begin{figure}[h]
\begin{centering}
    \mbox{
    \subfigure{\scalebox{.38}{\includegraphics{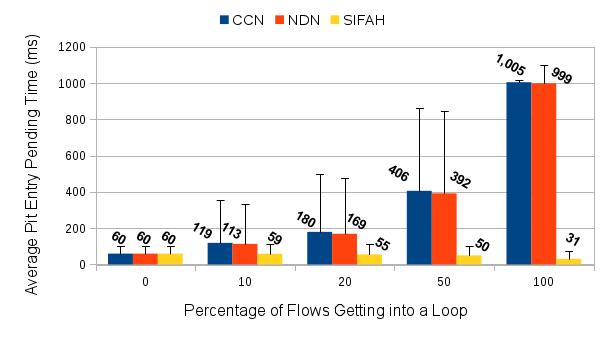}}}
      }
\vspace{-0.2in}
   \caption{Average PIT entry pending time for CCN, NDN, and SIFAH}
   \label{pending-time}
\end{centering} 
\end{figure}

By contrast, the average PIT entry pending time in CCN and NDN increases dramatically with the percentage of  Interests that 
encounter  Interest loops. The percentage of Interests that traverse loops need not be large to have negative performance consequences. For the scenario in which  10\%  of the Interests 
encounter Interest  loops, we observe that the average PIT entry pending time increases dramatically in NDN and CCN, with the average PIT entry pending time being 113 msec, which is about twice the average PIT entry pending time in SIFAH.  

The results for NDN and CCN can be easily explained. CCN simply deletes and drops duplicate Interests, each Interest that encounters an Interest loop is discarded by the router that detects a duplicate Interest, and this action forces the corresponding PIT entries in the routers traversed by the Interest to remain in those PITs, until their PIT entry expiration timers expire. 
In NDN, Interest loops  can go undetected with aggregation and therefore no NACKs are sent in those cases. As a result, the time an Interest entry spends in the PIT  equals the PIT entry expiration time if the Interest traverses an undetected loop.
The results are almost the same for CCN and NDN. The reason for observing slightly lower values for NDN compared to CCN, is that some of the Interests for content in $prefix_2$ are not generated by C1 and C2 with sufficient time correlation to enable  Interest aggregation, which results in detection of Interest loops in NDN and Interests being discarded in CCN.

The PIT entry pending times in NDN and CCN are many orders of magnitude larger for Interests that traverse undetected Interest loops. This is unavoidable,  given that the PIT entry pending time is  proportional to a PIT entry expiration time, which by design must be set conservatively to values that are far longer than average round-trip times between consumers and producers. 
In the simulations, the PIT entry expiration time is just one second.

\subsubsection{Impact on PIT Size}

Figure ~\ref{PitSize} shows the average size of PIT tables in terms of number of entries for a router included in Interest  flows for five different scenarios comparing CCN, NDN, and SIFAH. CCN, NDN and SIFAH have exactly the same PIT size in the absence of Interest loops, which is expected. As the  percentage of Interests that  encounter loops increases, the average number of entries in the PITs increases dramatically for CCN and NDN. For the case in which  only 10\%  of Interests encounter loops,  the number of entries doubles in NDN and CCN compared to SIFAH. For the case in which  100\% of Interests encounter loops, the average number of PIT entries in CCN and NDN is 1889 and 1884, respectively,  while the number of  PIT entries in SIFAH 
actually decreases. 

\begin{figure}[h]
\begin{centering}
    \mbox{
    \subfigure{\scalebox{.38}{\includegraphics{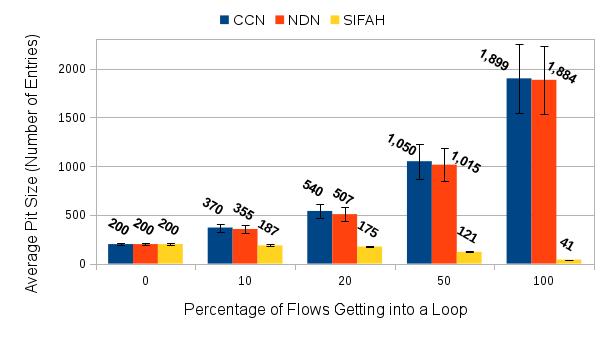}}}
      }
       \vspace{-0.2in}
   \caption{Average PIT table size for CCN, NDN, and SIFAH}
   \label{PitSize}
\end{centering} 
\end{figure} 

The reason for the decrease in average number of PIT entries for SIFAH as the percentage of Interests that encounter  loops increases is a consequence of the shorter round-trip times between the consumers submitting Interests and the routers sending NACKs  compared to the   round-trip times of paths  to the producers of requested content.

\subsubsection{Impact on Round-Trip Times}

Figure ~\ref{RTT} shows the average round-trip time (RTT) for all five scenarios for CCN, NDN and SIFAH. In the simulation experiments, the round-trip time is considered to be the time elapsed  from the instant when an  Interest is first sent to the instant when an NDO message or a NACK is received by the consumer who created the Interest.  

For the case of no loops, CCN, NDN, and SIFAH have the same average RTT. 
When the percentage of Interests traversing loops is 10\%, the average RTT in CCN and NDN increases to almost two times the average RTT in SIFAH, and some Interests have much larger RTTs than the average. As the percentage of Interests that loop increases, the average RTT becomes proportional to the PIT entry expiration time, which is to be expected.
The average RTT in SIFAH decreases as more Interests traverse loops, which is a result of the shorter RTTs between consumers and routers sending the NACKs.

\begin{figure}[h]
\begin{centering}
    \mbox{
    \subfigure{\scalebox{.38}{\includegraphics{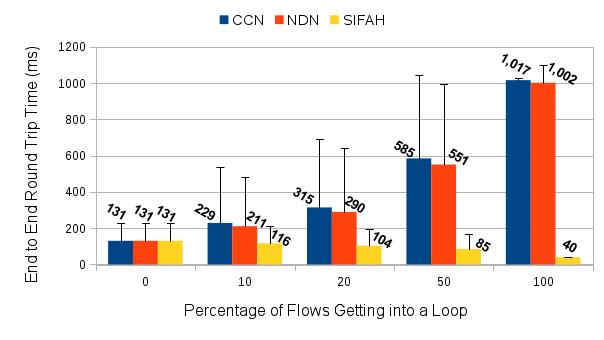}}}
      }
\vspace{-0.2in}
   \caption{Average round trip time (RTT)  for CCN, NDN, and SIFAH}
   \label{RTT}
\end{centering} 
\end{figure} 

\subsection{Design Implications}

The simulation experiments we have presented are meant only to help illustrate the negative impact of undetected Interest loops when they occur, rather than to provide representative scenarios of the performance of 
Interest-based forwarding strategies in large networks. Our results illustrate that  loops in FIBs need not be long lasting or impact a large percentage of Interest to cause the number of  stored PIT entries and end-to-end delays to increase quickly.

As we have shown, the PIT storage requirements  for SIFAH are smaller than those for the original CCN and NDN forwarding strategies. Thus, SIFAH is more efficient than CCN and NDN even in the absence of Interest loops.
Given that SIFAH is so easy  to implement in the context of CCN and NDN, it makes practical sense to eliminate the current 
practice in NDN and CCN of  attempting to detect Interest loops by the matching of nonces 
and Interest names, which does not work.   

\section{Conclusions}

We  showed that the forwarding strategies in NDN and the original CCN architectures  may fail to detect Interest loops when they occur, and that a correct forwarding strategy that supports Interest aggregation cannot be designed simply by  identifying each Interest uniquely and deciding that there is an Interest loop based on the matching of Interest names and nonces. 
 
We  introduced the Strategy for Interest  Forwarding and Aggregation with Hop-counts  (SIFAH). It  is the first Interest-based forwarding strategy shown to be correct in the presence of Interest loops, Interest aggregation, faults, and the forwarding of Interests over multiple paths. SIFAH operates by requiring that FIBs store the next hops and the hop count through such hops to named content, and by having each Interest state the name of the content requested and the hop count from the relaying router to the content. 

We showed that SIFAH incurs less storage overhead than using nonces to identify Interests. We also showed that, if NDN or the original CCN design  is used in a network,  the number of PIT entries and end-to-end delays perceived by consumers can increase substantially with just a fraction of Interests traversing undetected loops.  Although our  simulation experiments assumed a very small network,  our results provide sufficient insight on the negative effects of undetected Interest loops in NDN and the original CCN design.

This  work is just a first step in  the definition of correct Interest-based  forwarding strategies, and it is applicable to any Interest retransmission approach. For simplicity, we assumed that content consumers are in charge of Interest retransmissions and that routers do not provide local repair of Interests after receiving NACKs or detecting link failures. The design of an efficient Interest retransmission strategy and determining whether Interest retransmissions by routers improves performance are arguably the most important next steps.  However, SIFAH provides the necessary foundation to define a correct retransmission strategy, because it guarantees that each Interest results in an NDO message or a NACK being sent to the consumer who originated the Interest.

More work is also needed to understand the performance of SIFAH in large networks, the effect of PIT entry expiration  timers on performance, the effect of load balancing of Interests over multiple available routes to content, the impact of local repairs in Interest forwarding,  and the performance implications of the interaction between
SIFAH  and a routing protocol that guarantees loop-free routing tables (and hence FIBs) at all times \cite{dcr, icnp14,  dcr-mcast} compared to one that does not \cite{nlsr}.




\begin{thebibliography}{99}
 \vspace{0.01in}

\bibitem{ndnSIM}
A. Afanasyev, I. Moiseenko, and L. Zhang, ``ndnSIM: NDN Simulator for ns-3'', {\em University of California, Los Angeles, Tech. Rep}, 2012.


\bibitem{icn-survey1}
B. Ahlgren et al., ``A Survey of Information-centric Networking,"  
{\em IEEE Commun. Magazine}, July 2012, pp. 26--36.


\bibitem{alva}
J. Behrens and J.J. Garcia-Luna-Aceves, ``Hierarchical Routing Using Link Vectors,"
{\em Proc. IEEE INFOCOM `98}, April 1998.

\bibitem{ccnx}
Content Centric Networking Project (CCN) [online]. \\
http://www.ccnx.org/releases/latest/doc/technical/

\bibitem{caching}
A. Dabirmoghaddam et al., ``Understanding Optimal Caching and Opportunistic Caching at "The Edge'' of Information-Centric Networks,"
{\em Proc. ACM ICN `14},  Sept.  2014.

\bibitem{diffuse}
E.W. Dijkstra and C.S. Scholten ``Termination Detection for Diffusing Computations," {\em Information Processing Letters}, Vol. 11, No. 1, 1980.

\bibitem{dijkstra}
E.W. Dijkstra, W. Feijen, and A.J.M. van Gasteren, ``Derivation of a Termination Detection Algorithm for Distributed Computations," 
{\em Information Processing Letters},  Vol. 16, No. 5, 1983.


\bibitem{dual}
J.J. Garcia-Luna-Aceves,
``A Unified Approach to Loop-Free Routing Using Distance Vectors or Link States,"
{\em Proc. ACM SIGCOMM `89}, Aug. 1989.

\bibitem{dcr}
J.J. Garcia-Luna-Aceves, ``Name-Based Content Routing in Information Centric Networks Using Distance Information,"
{\em Proc. ACM ICN  `14},  Sept.  2014.

\bibitem{icnp14}
J.J. Garcia-Luna-Aceves, ``Routing to Multi-Instantiated Destinations: Principles and Applications,"
{\em IEEE ICNP `14}, Oct. 2014.

\bibitem{dcr-mcast}
J.J. Garcia-Luna-Aceves, ``Efficient Multi-Source Multicasting in Information Centric Networks,"
{\em Proc. IEEE CCNC  `15},  Jan.  2015.

\bibitem{diffusion}
C. Intanagonwiwat, R. Govindan, and D. Estrin, ``Directed Diffusion: A Scalable and Robust Communication Paradigm for Sensor Networks,"
{\em Proc. ACM MobiCom `00}, 2000.

\bibitem{ccn}
V. Jacobson et al., ``Networking Named Content," {\em Proc. IEEE CoNEXT `09}, Dec. 2009.	


\bibitem{nlsr}
A.K.M. Mahmudul-Hoque et al., ``NSLR: Named-Data Link State Routing Protocol," {\em Proc. ACM ICN `13}, 2013.


\bibitem{dtd}
J. Matocha and T. Camp, ``A Taxonomy of Distributed Termination Detection Algorithms,"
{\em Journal of Systems and Software}, 1998. 

\bibitem{ndn}
NDN Project [online]. http://www.named-data.net/

\bibitem{icnp98}
M. Spohn and J.J. Garcia-Luna-Aceves, ``Scalable Link-State Internet Routing,"
{\em Proc. IEEE ICNP `98}, Oct. 1998.

\bibitem{direct}
I. Solis and J.J. Garcia-Luna-Aceves, ``Robust Content Dissemination in Disrupted Environments,'' 
{\em Proc. ACM CHANTS `08}, Sept. 2008.

\bibitem{vutukury}
S. Vutukury and J.J. Garcia-Luna-Aceves, ``A Simple Approximation to Minimum-Delay Routing," {\em Proc. ACM SIGCOMM `99}, Aug. 1999.

\bibitem{icn-survey2}
G. Xylomenos et al., ``A Survey of Information-centric Networking Research,"  {\em IEEE Communication Surveys and Tutorials}, July 2013.


\bibitem{ndn-fw}
C. Yi et al., ``Adaptive Forwarding in Named Data Networking,"
{\em ACM CCR}, Vol. 42, No. 3, July 2012.

\bibitem{ndn-fw2}
C. Yi et al., ``A Case for Stateful Forwarding Plane,"
{\em Computer Communications}, pp. 779-791, 2013.


\bibitem{dual2}
W.T. Zaumen and J.J. Garcia-Luna-Aceves, ``Dynamics of Distributed Shortest-Path Routing Algorithms," 
{\em Proc. ACM SIGCOMM `91}, Sept. 1991.

\bibitem{dual-patent}
W.T. Zaumen and J.J. Garcia-Luna-Aceves, 
``System for Maintaining Multiple Loop-free Paths between Source Node and Destination Node in Computer Network,"
US Patent 5,881,243, 1999.

\bibitem{ndn-paper}
L. Zhang et al., ``Named Data Networking," {\em ACM SIGCOMM Computer Communication Review}, Vol. 44, No. 3, July 2014.

\end{thebibliography}
\end{document}